\documentclass{llncs}
\usepackage{amsmath}
\usepackage{listings}
\usepackage{amsfonts}
\usepackage{amsmath}
\usepackage{graphicx}
\usepackage{courier}
\usepackage{algorithmic}
\usepackage[table,xcdraw]{xcolor}
\usepackage{float}
\usepackage{hyperref}
\usepackage{mathtools}
\usepackage[framemethod=TikZ]{mdframed}

\newcommand{\viable}{{\mathsf {Viable}}}
\newcommand{\reachable}{{\mathsf {Reachable}}}
\newcommand{\extend}{{\mathsf {Extend}}}
\newcommand{\basecheck}{{\mathsf {BaseCheck}}}
\newcommand{\extendcheck}{{\mathsf {ExtendCheck}}}
\mathchardef\mhyphen="2D

\newenvironment{requirement}
{\vspace{0.05in}
 \begin{mdframed}[roundcorner=10pt,backgroundcolor=gray!20]}
{\end{mdframed}}

\begin{document}

\title{Towards Realizability Checking of Contracts using Theories}%
\author{Andrew Gacek\inst{1}, Andreas Katis\inst{2}, Michael W. Whalen\inst{2}, John Backes\inst{1}, Darren Cofer\inst{1}}%
\institute{Rockwell Collins Advanced Technology Center\\
400 Collins Rd. NE, Cedar Rapids, IA, 52498, USA\\
\email{\{andrew.gacek,john.backes,darren.cofer\}@rockwellcollins.com}
 \and
Department of Computer Science and Engineering,\\
 University of Minnesota, 200 Union Street, Minneapolis, MN 55455,USA\\
\email{katis001@umn.edu, whalen@cs.umn.edu}
}
\date{}%

\maketitle

\begin{abstract}
{\em Virtual integration} techniques focus on building architectural models of systems that can be analyzed early in the design cycle to try to lower cost, reduce risk, and improve quality of complex embedded systems.  Given appropriate architectural descriptions and compositional reasoning rules, these techniques can be used to prove important safety properties about the architecture prior to system construction.  Such proofs build from ``leaf-level'' assume/guarantee component contracts through architectural layers towards top-level safety properties.  The proofs are built upon the premise that each leaf-level component contract is {\em realizable}; i.e., it is possible to construct a component such that for any input allowed by the contract assumptions, there is some output value that the component can produce that satisfies the contract guarantees.
Without engineering support it is all too easy to write leaf-level components that can't be realized.  Realizability checking for propositional contracts has been well-studied for many years, both for component synthesis and checking correctness of temporal logic requirements.  However, checking realizability for contracts involving infinite theories is still an open problem.  In this paper, we describe a new approach for checking realizability of contracts involving theories and demonstrate its usefulness on several examples.

\end{abstract}

\vspace{-0.1in}
\section{Introduction}
\vspace{-0.1in}

In the recent years, {\em virtual integration} approaches have been proposed as a means to lower cost and improve quality of complex embedded systems.  These approaches focus on building architectural models of systems that can be analyzed prior to construction of component implementations.   The objective is to discover and resolve problems early during the design and implementation phases when cost impact is lower.  Several architecture description languages such as AADL~\cite{SAE:AADL}, SysML~\cite{Friedenthal08:sysml}, and AUTOSAR~\cite{AUTOSAR} are designed to support such an engineering process, and there has been significant effort to analytically determine system performance~\cite{Gomez11:AADL,Bozzano:2011:AADL}, fault tolerance~\cite{Bozzano:2011:AADL}, security~\cite{Apvrille13:security}, and safety~\cite{Bozzano14:AADL} using these techniques.

In an ongoing effort at Rockwell Collins and The University of Minnesota, we have been pursuing virtual integration using compositional proofs of correctness.  The idea is to support hierarchical design and analysis of complex system architectures and co-evolution of requirements and architectures at multiple levels of abstraction~\cite{Whalen13:WhatHow:TwinPeaksIEEESoftware}.  This was based on two observations about software development for commercial aircraft: first, that component-level errors are relatively rare and that most problems occur during integration~\cite{rushby2011new}, and second, that requirements specifications often contain significant numbers of omissions or errors~\cite{Miller06:Shalls} that are at the root of many of the integration problems.  Specifically, the problem involves demonstrating {\em satisfaction arguments}~\cite{Hammond01:WiW}, i.e., that the requirements allocated to components and the architecture connecting those components is sufficient to guarantee the system requirements.  We have created the AGREE
reasoning framework~\cite{NFM2012:CoGaMiWhLaLu} to support compositional assume/guarantee contract reasoning
over system architectural models written in AADL.

Such proof systems build from ``leaf-level'' assume/guarantee component contracts through architectural layers towards proofs of top-level safety properties.  The soundness of the argument is built upon the premise that each leaf-level component contract is {\em realizable}; i.e., it is possible to construct a component such that for any input allowed by the contract assumptions, there is some output value that the component can produce that satisfies the contract guarantees.

Unfortunately, without engineering support it is all too easy to write leaf-level components that can't be realized.  When applying our tools in both industrial and classroom settings, this issue has led to incorrect compositional ``proofs'' of systems; in fact the goal of producing a compositional proof can lead to engineers modifying component-level requirements such that they are no longer possible to implement. In order to make our virtual integration approach reasonable for practicing engineers, tool support must be provided to check whether components are {\em realizable}.

Realizability checking for propositional contracts has been well-studied for many years (e.g.,~\cite{Pnueli89,Bohy12,Hamza10,Chatterjee07}), both for component synthesis and checking correctness of temporal logic requirements.  Checking realizability for contracts involving theories, on the other hand, is still an open problem.  In this paper, we describe a new approach for checking realizability of contracts involving theories and demonstrate its usefulness on several examples.  Our approach is similar to k-induction over quantified formulas.  We describe two algorithms.  The first is sound for both proofs and counterexamples, but computationally intractable.  The second algorithm is not sound for counterexamples (i.e., it may return a `false counterexample' to a problem that is in fact realizable), but we have found it fast and accurate in practice.

The rest of the paper is structured as follows.  In Section~\ref{sec:motivation} we will describe our motivation and an example to illustrate realizability, and will define
realizability formally in Section~\ref{sec:realizability}.  We next describe two algorithms for checking realizability in Section~\ref{sec:algorithm}, our implementation in the AGREE tool suite in Section~\ref{sec:implementation}, and our experience using the realizability check in Section~\ref{sec:case_studies}.  Section~\ref{sec:related_work} describes related work and Section~\ref{sec:future_work} concludes.

\vspace{-0.1in}
\section{Motivation and Example}
\vspace{-0.1in}
\label{sec:motivation}
We have been pursuing a {\em proof-based virtual integration} approach for building complex systems using the architecture description language AADL~\cite{SAE:AADL} and the AGREE compositional reasoning system~\cite{NFM2012:CoGaMiWhLaLu}.  We have demonstrated the effectiveness of the approach on a variety of industrial-scale systems, including the software controller for a patient-controlled analgesia (PCA) infusion pump~\cite{hilt2013}, a dual flight-guidance system~\cite{NFM2012:CoGaMiWhLaLu}, and several more recent models, such as a quad-redundant flight control system and a quadcopter control system.  We are using this approach on the DARPA HACMS program to build secure vehicles and to demonstrate how to apply virtual integration on industrial scale systems to facilitate technology transfer.

As part of the HACMS project, we attempted a feasibility test via a classroom exercise.  We used the AADL and AGREE tools in a class assignment in a graduate-level software architecture class.  
The students were organized into six teams of four students.  Each team was asked to specify the control software for a simplified microwave oven in AADL using a virtual integration approach.  The software was split into two subsystems: one for controlling the heating element and another for controlling the display panel, with several requirements for each subsystem.  The goal was to formalize these component-level requirements and use them to prove three system-level safety requirements.

The results of the initial experiment were sobering.  All student groups were
able to prove the system-level requirements starting from formalizations of the component requirements.  Unfortunately, in many cases, the proofs succeeded because the components were incorrectly specified.  In fact, only one of the teams had written component-level requirements that could be implemented.  The other teams had requirements which were inconsistent under certain input conditions.  For example, one team produced the following informal component-level requirements:
\begin{requirement}
\textbf{Microwave-1} - While the microwave is in cooking mode,
{\sf seconds\_to\_cook} shall decrease.
\label{req:mic1}
\end{requirement}
\begin{requirement}
\textbf{Microwave-2} - If the display is quiescent (no buttons pressed) and the
keypad is enabled, the {\sf seconds\_to\_cook} shall not change.
\label{req:mic2}
\end{requirement}
and then produced the following formalized
requirements\footnote{We have translated this property and others from the
  higher level AGREE syntax into a two-state form that is used
  throughout this paper.}:
\begin{align*}
\mathbf{guarantee} &\colon \mathsf{is\_cooking'}
\Rightarrow \mathsf{seconds\_to\_cook'} \leq \mathsf{seconds\_to\_cook} -1 \\[5pt]
\mathbf{guarantee} &\colon
  (\neg \mathsf{any\_digit\_pressed} \land \mathsf{keypad\_enabled})
\Rightarrow \\
&\hspace{2em} \mathsf{seconds\_to\_cook'} = \mathsf{seconds\_to\_cook}
\end{align*}
These formalized guarantees fail to avoid the conflict in the {\sf
  seconds\_to\_cook} variable between the
~\hyperref[req:mic1]{Microwave-1} and~\hyperref[req:mic2]{Microwave-2}
requirements, as they cannot be both satisfied in a case where the
microwave is cooking and the keypad is enabled. This error was not
caught despite an analysis built into an early version of AGREE that
checks contracts for {\em consistency}, i.e., whether the conjunction
of a system's guarantees is satisfiable. We realized that consistency
checking does not actually provide a trustworthy answer because it
only checks whether the system works in {\em some} external
environment, not in {\em all} environments. Realizability checking
determines whether or not the component works in all input
environments that satisfy the component assumptions.

From this experience, we decided that realizability checking was
necessary for successful tech transfer of a virtual integration
approach. The analysis was not only necessary for classroom settings.
We also found problems with component-level requirements in two of our
large-scale analysis efforts. Further, existing approaches for
checking realizability do not allow predicates over infinite theories
such as integers and reals, which are native to our AGREE contracts.

In the following sections, we formally define realizability over
transition systems, as well as algorithms for checking realizability
over infinite-state systems that are efficient and accurate in
practice. A machine-checked formalization of the definitions and
proofs in Coq can be found in a companion paper \cite{Katis:machine}.



\newcommand{\state}{\mathsf{state}}
\newcommand{\inputty}{\mathsf{input}}
\newcommand{\bool}{\mathsf{bool}}

\vspace{-0.1in}
\section{Realizability}
\label{sec:realizability}
\vspace{-0.1in}

We assume the types $\state$ and $\inputty$ for states and inputs. We use
$s$ for variables of type $\state$ and $i$ for variables of type
$\inputty$. State represents both internal state and external outputs. A
transition system is a pair $(I, T)$ where $I : \state \to \bool$ holds
on the initial states states and $T : \state \times \inputty \times \state
\to \bool$ holds on $T(s, i, s')$ when the system can transition from
state $s$ to state $s'$ on receipt of input $i$. We assume the usual
notion of path with respect to a transition relation.

A contract specifies the desired behavior of a transition system. A
contract is a pair $(A, G)$ of an assumption and a guarantee. The
assumption $A : \state \times \inputty \to \bool$ specifies for a given
system state which inputs are valid. The guarantee $G$ is a pair
$(G_I, G_T)$ of an initial guarantee and a transitional guarantee. The
initial guarantee $G_I : \state \to \bool$ specifies which states the
system may start in, that is, the possible initial internal state and
external outputs. The transitional guarantee $G_T : \state \times
\inputty \times \state \to \bool$ specifies for a given state and input
what states the system may transition to.

We now define what it means for a transition system to realize a
contract. This requires that the system respects the guarantee for
inputs which satisfying the contract. Moreover, the system must always
remain responsive with respect to inputs that satisfying the
assumptions. In order to make this definition precise, we first need
to define which system states are reachable given some assumptions on
the system inputs.

\begin{definition}[Reachable with respect to assumptions]
\label{def:reachable}
Let $(I, T)$ be a transition system and let $A : \state \times \inputty
\to \bool$ be an assumption. A state of $(I, T)$ is reachable with
respect to $A$ if there exists a path starting in an initial state and
eventually reaching $s$ such that all transitions satisfying the
assumptions. Formally, $\reachable_A(s)$ is defined inductively by
\begin{equation*}
  \reachable_A(s) = I(s) \lor \exists s_{\mathsf{prev}}, i.~ \reachable_A(s_{\mathsf{prev}}) \land
  A(s_{\mathsf{prev}}, i) \land T(s_{\mathsf{prev}}, i, s)
\end{equation*}
\end{definition}

\begin{definition}[Realization]
\label{def:realization}
  A transition system $(I, T)$ is a realization of the contract $(A,
  (G_I, G_T))$ when the following conditions hold
  \begin{enumerate}
  \item $\forall s.~ I(s) \Rightarrow G_I(s)$
  \item $\forall s, i, s'.~ \reachable_A(s) \land A(s, i) \land T(s, i,
    s') \Rightarrow G_T(s, i, s')$
  \item $\exists s.~ I(s)$
  \item $\forall s, i.~ \reachable_A(s) \land A(s, i) \Rightarrow
    \exists s'.~ T(s, i, s')$
  \end{enumerate}
\end{definition}

The first two conditions in Definition~\ref{def:realization} ensure
that the transition system respects the guarantees. The second two
conditions ensure that the system is non-trivial and responsive to all
valid inputs.

\begin{definition}[Realizable]
\label{def:realizable}
A contract is realizable if there exists a transition system which is
a realization of the contract.
\end{definition}

Definitions~\ref{def:realization} and \ref{def:realizable} are useful
for directly defining realizability, but not very useful for checking
realizability. We now develop an equivalent notion which is more
suggestive and amenable to checking. This is based on a notion called
{\em viability}. Intuitively, a state is viable with respect to a
contract if being in that state does not doom a realization to
failure. We can capture this notion without reference to any specific
realization, because condition 2 in the definition of realization
tells us that $G_T$ is an over-approximation of any $T$.

\begin{definition}[Viable] A state $s$ is {\em viable} with respect to
  a contract $(A, (G_I, G_T))$, written $\viable(s)$, if $G_T$ can
  keep responding to valid inputs forever, starting from $s$.
  Informally, one can say that a state $s$ is viable if it satisfies
  the infinite formula:
\begin{equation*}
\forall i_1.~ A(s, i_1) \Rightarrow \exists s_1.~ G_T(s, i_1, s_1) \land
\forall i_2.~ A(s_1, i_2) \Rightarrow \exists s_2.~ G_T(s_1, i_2, s_2)
\land \forall i_3.~ \cdots
\end{equation*}
Formally, viability is defined coinductively by the following equation
\begin{equation*}
\viable(s) = \forall i.~ A(s, i) \Rightarrow \exists s'.~ G_T(s, i, s') \land \viable(s')
\end{equation*}
\end{definition}

\begin{theorem}[Alternative realizability]
  \label{thm:alt-realizable}
  A contract $(A, (G_I, G_T))$ is realizable if and only if $\exists s.~ G_I(s)
  \land \viable(s)$.
\end{theorem}
\begin{proof}
For the ``only if'' direction the key lemma is $\forall s.~ \reachable_A(s)
\Rightarrow \viable(s)$. This lemma is proved by coinduction and
follows directly from conditions 2 and 4 of
Definition~\ref{def:realization}. Then by conditions 1 and 3 we have
some state $s$ such that $I(s)$ and $G_I(s)$. Thus $\reachable_A(s)$
holds and applying the lemma we get $G_I(s) \land \viable(s)$.

For the ``if'' direction, let $s_0$ be such that $G_I(s_0)$ and
$\viable(s_0)$. Define $I(s) = (s = s_0)$ and $T(s,i,s') = G_T(s, i,
s') \land \viable(s')$. Conditions 1, 2, and 3 of
Definition~\ref{def:realization} are clearly satisfied. Condition 4
follows from the observation that $\forall s.~ \reachable_A(s)
\Rightarrow \viable(s)$ and from the definition of viability.
\end{proof}



\vspace{-0.1in}
\section{An Algorithm for Checking Realizability}
\label{sec:algorithm}
\vspace{-0.1in}

In this section we develop two versions of an algorithm for
automatically checking the realizability of a contract. The first
version is based on Theorem~\ref{thm:alt-realizable} together with
under- and over-approximations of viability. An over-approximation is
useful to show that a contract is not viable, while an
under-approximation is useful to show that a contract is viable. The
second version of the algorithm follows from the mitigating the
intractability of the first version.

We first define an over-approximation of viability called {\em finite
  viability} based on a finite unrolling of the definition of
viability. Because this is an over-approximation, if a contract does
not have an initial state which is finitely viable, then the contract
is not viable. We formalize this when we prove the correctness of the
realizability algorithm.

\begin{definition}[Finite viability]
A state $s$ is viable for $n$ steps, written $\viable_n(s)$ if $G_T$
can keep responding to valid inputs for at least $n$ steps. That is,
\begin{multline*}
\forall i_1.~ A(s, i_1) \Rightarrow \exists s_1.~ G_T(s, i_1, s_1) \land~ \\
\forall i_2.~ A(s_1, i_2) \Rightarrow \exists s_2.~ G_T(s_1, i_2, s_2)
\land \cdots \land~ \\
\forall i_n.~ A(s_{n-1}, i_n) \Rightarrow \exists s_n.~ G_T(s_{n-1}, i_n, s_n)
\end{multline*}
All states are viable for 0 steps.
\end{definition}

We next define an under-approximation of viability based on {\em
  one-step extension}. This notion looks if $G_T$ can respond to valid
inputs given a finite historical trace of valid inputs and states.

\begin{definition}[One-step extension]
A state $s$ is extendable after $n$ steps, written $\extend_n(s)$, if
any valid path of length $n$ from $s$ can be extended in response to
any input. That is,
\begin{multline*}
\forall i_1, s_1, \ldots, i_n, s_n.\\ A(s, i_1) \land G_T(s, i_1, s_1)
\land \cdots \land
A(s_{n-1}, i_n) \land G_T(s_{n-1}, i_n, s_n)
\Rightarrow \\
\forall i.~ A(s_n, i) \Rightarrow \exists s'.~ G_T(s_n, i, s')
\end{multline*}
\end{definition}

We now use these two notions to formally define our realizability
algorithm. The core of the algorithm is based on two checks called the
{\em base} and {\em extend} check.

\begin{definition}[Realizability Algorithm]
Define the checks:
  \label{def:algorithm1}
  \begin{align*}
    \basecheck(n) &= \exists s.~ G_I(s) \land \viable_n(s) \\
    \extendcheck(n) &= \forall s.~ \extend_n(s)
  \end{align*}
The following algorithm checks for realizability or unrealizability
of a contract.
\newpage
\begin{algorithmic}
  \FOR{$n=0$ to $\infty$}
    \IF{\NOT $\basecheck(n)$}
      \RETURN{``unrealizable''}
    \ELSIF{$\extendcheck(n)$}
      \RETURN{``realizable''}
    \ENDIF
  \ENDFOR
\end{algorithmic}
\end{definition}

\begin{theorem}[Soundness of ``unrealizable'' result]
If $\exists n.~ \neg\basecheck(n)$ then the contract is not
realizable.
\end{theorem}
\begin{proof}
First we show $\forall s, n.~ \viable(s) \Rightarrow \viable_n(s)$ by
induction on $n$. The result then follows from
Theorem~\ref{thm:alt-realizable}.
\end{proof}

\begin{theorem}[Soundness of ``realizable'' result]
If $\exists n.~ \basecheck(n) \land \extendcheck(n)$ then contract is
realizable.
\end{theorem}
\begin{proof}
First we show how $\extend_n(s)$ can be used to shift $\viable_n(s)$
forward. The following is proved by induction on $n$.
\begin{equation*}
  \forall s, n, i.~ \extend_n(s) \land \viable_n(s) \land A(s, i)
  \Rightarrow \exists s'.~ G_T(s, i, s') \land \viable_n(s')
\end{equation*}
Using this lemma we can show the following by coinduction.
\begin{equation*}
  \forall s, n.~ \viable_n(s) \land \extendcheck(n) \Rightarrow
  \viable(s)
\end{equation*}
The result then follows from Theorem~\ref{thm:alt-realizable}.
\end{proof}

\begin{corollary}[Soundness of Realizability Algorithm]
  The Realizability Algorithm is sound.
\end{corollary}

Due to the approximations used to define the base and extends check,
the algorithm is incomplete. The following two examples show how both
realizable and unrealizable contracts may send the algorithm into an
infinite loop.

\begin{example}[Incompleteness of ``realizable'' result]
\label{ex:realizable-incomplete}
Suppose the type $state$ is integers. Consider the contract:
\begin{align*}
  A(s, i) &= \top & G_I(s) &= \top & G_T(s, i, s') &= (s \neq 0)
\end{align*}
This contract is realizable by, for example, a system that starts in
state $1$ and always transitions into the same state. Yet, for all
$n$, $\extendcheck(n)$ fails since one can take a path of length $n$
which ends at state $0$. This path cannot be extended.
\end{example}

\begin{example}[Incompleteness of ``unrealizable'' result]
\label{ex:unrealizable-incomplete}
Suppose the type $state$ is integers. Consider the contract:
\begin{align*}
  A(s, i) &= \top & G_I(s) &= (s \geq 0) & G_T(s, i, s') &= (s' = s - 1
  \land s' \geq 0)
\end{align*}
This contract is not realizable since in any realization the state $0$
would be reachable, but the contract does not allow a transition from
state $0$. However, $\basecheck(n)$ holds for all $n$ by starting in
state $s = n$.
\end{example}

Implementing this algorithm requires a way of automatically checking
the formulas $\basecheck(n)$ and $\extendcheck(n)$ for validity. This
can be done in an SMT-solver that supports quantifiers over the
language the contract is expressed in. Checking $\extendcheck(n)$ is
rather nice in this setting since it has only a single quantifier
alternation. Moreover, using an incremental SMT-solver one can reuse
much of the work done to check $\extendcheck(n)$ to also check
$\extendcheck(n+1)$. However, $\basecheck(n)$ is problematic.
First, it has $2n$ quantifier alternations which puts even small cases
outside the reach of modern SMT-solvers. Second, the quantifiers make
it impractical to reuse the results of $\basecheck(n)$ in checking
$\basecheck(n+1)$. Finally, due to the quantifiers, a counterexample
to $\basecheck(n)$ would be difficult to relay back to the user. Thus
we need a simplification of $\basecheck(n)$ in order to make our
algorithm practical.

\begin{definition}[Simplified base check]
Define a simplified base check which checks that any path of length
$n$ from an initial state can be extended one step.
  \begin{equation*}
    \basecheck'(n) = \forall s.~ G_I(s) \Rightarrow \extend_n(s)
  \end{equation*}
\end{definition}

First, note that this check has a single quantifier alternation.
Second, this check can leverage the incremental features in an
SMT-solver to use the results of $\basecheck'(n)$ in checking
$\basecheck'(n+1)$. Finally, when this check fails it can return a
counterexample which is a trace of a system realizing the contract for
$n$ steps, but then becoming stuck. This provides very concrete and
useful feedback to system developers. The correctness of this check is
captured by the following theorem.

\begin{theorem}[One-way soundness of simplified base check]
  \label{thm:basecheck'}
  \begin{equation*}
    (\exists s.~ G_I(s)) \Rightarrow \forall n.~ (\forall k\leq n.~ \basecheck'(k)) \Rightarrow \basecheck(n)
  \end{equation*}
\end{theorem}
\begin{proof}
We first prove the following by induction on $n$:
\begin{equation*}
  \forall s, n.~ \extend_n(s) \land \viable_n(s) \Rightarrow \viable_{n+1}(s)
\end{equation*}
The final result follows using this and induction on $n$.
\end{proof}

Thus replacing $\basecheck(n)$ in the realizability algorithm with
$\basecheck'(n)$ preserves soundness of the ``realizability'' result.
However, because the implication in Theorem~\ref{thm:basecheck'} is
only in one direction, the algorithm is no longer sound for the
``unrealizable'' result. That is, it may return a counterexample
showing $n$ steps of a realization of the contract that gets into a
stuck state. The following example makes this point explicit.

\begin{example}
Consider again Example~\ref{ex:realizable-incomplete} where
the type $state$ is integers and the contract is:
\begin{align*}
  A(s, i) &= \top & G_I(s) &= \top & G_T(s, i, s') &= (s \neq 0)
\end{align*}
As before, this contract is easily realizable. However,
$\basecheck'(n)$ fails for all $n$ since it will consider a path
starting at state $n$ and transitioning $n$ steps to state $0$ where
no more transitions are possible.
\end{example}

The benefits of this second version of the algorithm outweigh its
costs. The cases where a contract is realizable, yet fails the
modified base check seems unlikely in practice. We have encountered
none in our case studies. Moreover, when a contract does spuriously
fail the simplified base check, it can almost always be rewritten into
a form which would pass.



\vspace{-0.1in}
\section{Implementation}
\label{sec:implementation}
\vspace{-0.1in}

We have built an implementation of the realizability algorithm as an
extension to JKind~\cite{jkind}, a re-implementation of the KIND model
checker~\cite{Hagen08:kind} in Java. Our tool is called JRealizability
and is packaged with the latest release of JKind. The model's behavior
is described in the Lustre language, which is the native input
language of JKind and is used as an intermediate language for the
AGREE tool suite.

We unroll the transition relation defined by the Lustre model into SMT
problems (one for the base check and another for the extend check)
which can be solved in parallel. We use the SMT-LIB Version 2 format
which most modern SMT solvers support. The most significant issue for
SMT solvers involves quantifier support, so we use the Z3 SMT
solver~\cite{DeMoura08:z3} which has good support for reasoning over
quantifiers and incremental search. The tool is often able to provide
an answer for models containing integer and real-valued variables very
quickly (in less than a second). Because of the use of quantifiers
over a range of theories, it is possible that for one of the checks,
Z3 returns {\tt unknown}; in this case, we discontinue analysis. In
addition, because our realizability check is incomplete, the tool
terminates analysis when either a timeout or a user-specified max
unrolling depth (default: 200) is reached. In this case we are able to
report how far the base check reached which may provide some
confidence in the realizability of the system.



\vspace{-0.1in}
\section{Case Studies}
\label{sec:case_studies}
\vspace{-0.1in}

As a part of testing the algorithm in actual components, we examined
three different cases: a quad-redundant flight control system, a medical
infusion pump, and a simple microwave controller. In this section,
we provide a brief description of each case study and summarize the
results in Table~\ref{table:case-studies} at the end of the section.

\vspace{-0.1in}
\subsection{Quad-Redundant Flight Control System}
\vspace{-0.1in}

We ran our realizability analysis on a Quad-Redundant Flight Control
System (QFCS) for NASA's Transport Class Model (TCM) aircraft
simulation. We were provided with a set of English language
requirements for the QFCS components and a description of the
architecture. We modeled the architecture in AADL and the component
requirements as assume/guarantee contracts in AGREE. As the name
suggests, the QFCS consists of four redundant Flight Control Computers
(FCCs). Each FCC contains components for handling faults and computing
actuator signal values. One of these components is the Output Signal
Analysis and Selection component (OSAS). The OSAS component is
responsible for determining the output gain for signals coming from
the control laws and going to the actuators. The output signal gain is
determined based on the number of other faulty FCCs or based on
failures within the FCC containing the OSAS component. The OSAS
component contains 17 English language requirements including the
following:
\begin{requirement}
\textbf{OSAS-S-170} -- If the local Cross Channel Data Link (CCDL) has
failed, OSAS shall set the local actuator command gain to 1 (one).
\label{req:osas-ccdl-fail}
\end{requirement}
\begin{requirement}
\textbf{OSAS-S-240} -- If OSAS has been declared failed by CCDL, OSAS
shall set the actuator command gain to 0 (zero).
\label{req:osas-fail}
\end{requirement}
We formalized these requirements using the following guarantees:
\begin{align*}
\mathbf{guarantee} &\colon \mathsf{ccdl\_failed} \Rightarrow
(\mathsf{fcc\_gain'} = 1) \\
\mathbf{guarantee} &\colon \mathsf{osas\_failed} \Rightarrow
(\mathsf{fcc\_gain'} = 0)
\end{align*}
These guarantees are contradictory in the case when the local CCDL has
failed and the local CCDL reports to the OSAS that the OSAS has
failed. This error eluded the engineers who originally drafted the
requirements as well as the engineers who formalized them. In this
case, there should be an assumption that if the CCDL has failed then
it will not report to the OSAS that the OSAS has failed. This was not
part of the original requirements. However, AGREE's realizability
analysis was able to identify the error and provide a counterexample.

\vspace{-0.1in}
\subsection{Medical Device Example}
\vspace{-0.1in}

Our realizability tool was also used to verify the realizability of
the components in the Generic Patient Controlled Analgesia infusion pump system
that was described in~\cite{Murugesan2014}.
The controller consists of six subcomponents that were given as input for the
tool to verify the requirements described inside. While five of the models were
proven to be realizable, a subtly incorrect requirement definition was found in
the contract for the controller's infusion manager.

\begin{requirement}
\textbf{GPCA-1} - The mode range of the controller shall be one of nine
different modes. If the controller is in one of the first two modes the
commanded flow rate shall be zero.
\label{req:gpca1}
\end{requirement}
\begin{equation}
\begin{aligned}
&\mathbf{guarantee} \colon \\
  &\hspace{2em}(\mathsf{IM\_OUT.Current\_System\_Mode'} \geq 0)~\land~ \\
  &\hspace{2em}(\mathsf{IM\_OUT.Current\_System\_Mode'} \leq 8)~\land~ \\
  &\hspace{2em}(\mathsf{IM\_OUT.Current\_System\_Mode'} = {0} \Rightarrow~ \\
      &\hspace{4em}\mathsf{IM\_OUT.Commanded\_Flow\_Rate'} = 0)\land~ \\
  &\hspace{2em}(\mathsf{IM\_OUT.Current\_System\_Mode'} = 1 \Rightarrow~ \\
      &\hspace{4em}\mathsf{IM\_OUT.Commanded\_Flow\_Rate'} = 0)
\end{aligned}
\label{guar:gpca1}
\end{equation}
\begin{requirement}
\textbf{GPCA-2} - Whenever the alarm subsystem has detected a high severity hazard, then Infusion Manager shall never infuse drug at a rate more than the specified Keep Vein Open rate.
\label{req:gpca}
\end{requirement}
\begin{equation}
\begin{aligned}
&\mathbf{guarantee} \colon \\
&\hspace{2em}(\mathsf{TLM\_MODE\_IN.System\_On'} \land~ \\
&\hspace{4em}\mathsf{ALARM\_IN.Highest\_Level\_Alarm'} = 3) \Rightarrow \\
&\hspace{2em}(\mathsf{IM\_OUT.Commanded\_Flow\_Rate'} =
  \mathsf{CONFIG\_IN.Flow\_Rate\_KVO'})
\end{aligned}
\label{guar:gpca2}
\end{equation}

\sloppypar
The erroneously defined guarantee~(\ref{guar:gpca2}) tries to assert that the $\mathsf{IM\_OUT.Commanded\_Flow\_Rate}$ to some (potentially non-zero) $\mathsf{Flow\_Rate\_KVO}$ if the alarm input is 3; however, this may occur when the $\mathsf{IM\_OUT.Current\_System\_Mode}$ is computed to be zero or one, in which case the flow rate is commanded to be 0.  While  discovering and fixing the problem was not difficult, the error was not discovered by the regular consistency check in AGREE.

\vspace{-0.1in}
\subsection{Microwave Assignment}
\vspace{-0.1in}

The realizability tool was used to check the contracts for the microwave models produced by the graduate student teams described in Section~\ref{sec:motivation} that provided the initial motivation for this work.  The microwave consists of two subsystems that manage the cooking element and display panel of the device. Table~\ref{table:case-studies} shows the corresponding results for each team, named as MT1, MT2, etc. While
every team but one managed to provide an implementable set of requirements for the microwave's mode
controller, there were several interesting cases involving the display control component.  For space reasons, we highlight only one here.

\begin{requirement}
\textbf{Microwave-1} - While the microwave is in cooking mode,
{\sf seconds\_to\_cook} shall decrease.
\label{req:mic1}
\end{requirement}

\begin{requirement}
\textbf{Microwave-3} - When the keypad is initially enabled, if no digits are
pressed, the value shall be zero.
\label{req:mic3}
\end{requirement}
Team 6 formalized these requirements as
\begin{align*}
&\mathbf{guarantee} \colon (\mathsf{cooking\_mode'} = 2) \Rightarrow
(\mathsf{seconds\_to\_cook'} = \mathsf{seconds\_to\_cook} - 1) \\
&\mathbf{guarantee} \colon (\lnot \mathsf{keypad\_enabled} \land \mathsf{keypad\_enabled'} \land
  \lnot \mathsf{any\_digit\_pressed'}) \Rightarrow \\
  & \ \ \ \ (\mathsf{seconds\_to\_cook'} = 0)
\end{align*}

\noindent In the counterexample provided, the state where the microwave is cooking ($\mathsf{cooking\_mode} = 2$) and no digit is pressed creates a conflict regarding which value is assigned to the
$\mathsf{seconds\_to\_cook}$ variable: should it decrease by one, or be
assigned to zero?  This counterexample is interesting because it indicates a missing assumption on the environment: the keypad is not enabled when the cooking mode is 2 (cooking).  Without this assumption about the inputs, the guarantees are not realizable.

\begin{table}[H]
\resizebox{\textwidth}{!}{%
\begin{tabular}{|c|c|c|c|c|}
\hline
\textbf{Case study} & \textbf{Model} & \textbf{Result} & \textbf{Time elapsed (seconds)} & \textbf{\begin{tabular}[c]{@{}c@{}}Base check depth \\ (\# of steps)\end{tabular}} \\ \hline
QFCS & FCS & realizable & 1.762 & 0 \\ \hline
QFCS & FCC & unrealizable & 0.981 & 1 \\ \hline
GPCA & Infusion Manager & unrealizable & 0.2 & 1 \\ \hline
GPCA & Alarm & realizable & 0.316 & 0 \\ \hline
GPCA & Config & realizable & 0.102 & 0 \\ \hline
GPCA & OutputBus & realizable & 0.201 & 0 \\ \hline
GPCA & System\_Status & realizable & 0.203 & 0 \\ \hline
GPCA & Top\_Level & realizable & 0.103 & 0 \\ \hline
MT 1 & Mode Control & realizable & 0.229 & 0 \\ \hline
MT 1 & Display Control & unrealizable & 0.207 & 1 \\ \hline
MT 2 & Mode Control & realizable & 0.202 & 0 \\ \hline
MT 2 & Display Control & unknown & 1000 (tool timeout) & 1 \\ \hline
MT 3 & Mode Control & realizable & 0.203 & 0 \\ \hline
MT 3 & Display Control & unrealizable & 0.202 & 1 \\ \hline
MT 4 & Mode Control & realizable & 0.202 & 0 \\ \hline
MT 4 & Display Control & unrealizable & 0.521 & 1 \\ \hline
MT 5 & Mode Control & unrealizable & 0.1 & 1 \\ \hline
MT 5 & Display Control & unrealizable & 0.222 & 1 \\ \hline
MT 6 & Mode Control & realizable & 0.201 & 0 \\ \hline
MT 6 & Display Control & unknown & 1000 (tool timeout) & 1 \\ \hline
\end{tabular}
}
\vspace{+1em}
\caption{Realizability checking results for case studies}
\label{table:case-studies}
\end{table}

\vspace{-3em}

Table~\ref{table:case-studies} contains the exact results that were
obtained during the three case studies. Every ``realizable'' result
was determined to be correct since an implementation was produced for
each of the components analyzed, ensuring the accuracy of the tool.
Every contract that was identified as ``unrealizable'' was manually
confirmed to be unrealizable, i.e., there were no spurious results.
Additionally, the number of steps that the base check required to
provide a final answer was not more than one, with the unknown results
being particularly interesting, as the tool timed out before the
solver was able to provide a concrete answer. This shows that there
is still work to be done in terms of the algorithm's scalability, as
well as an efficient way to eliminate quantifiers, making the solving
process easier for Z3.



\vspace{-0.1in}
\section {Related Work}
\label{sec:related_work}
\vspace{-0.1in}

The idea of $realizability$ has been the subject of intensive study. 
Gunter et al. refer to it using the term $\textit{relative
consistency}$ in~\cite{Gunter00}, while Pnuelli and Rosner use the term
$\textit{implementability}$ in~\cite{Pnueli89} to refer to the problem of synthesis for
propositional LTL. Additionally, the authors in~\cite{Pnueli89} proved that the
lower-bound time complexity of the problem is doubly exponential, in the worst
case. In the following years, several techniques were introduced to deal with
the synthesis problem in a more efficient way for subsets of propositional LTL
\cite{Klein10}, simple LTL formulas (\cite{Bohy12}, \cite{Tini03}), as well as
in a component-based approach \cite{Chatterjee07} and specifications based on
other temporal logics (\cite{Benes12}, \cite{Hamza10}), such as SIS
\cite{Aziz95}.
Finally, an interesting and relevant work has been done regarding the solution
to the controllability problem using in \cite{micheli_aaai_2012}
\cite{micheli_cp_2012} and \cite{micheli_constraints_2014}, which involves the
decision on the existence a strategy that assigns certain values to a set of
controllable activities, with respect to a set of uncontrollable ones.

Recent work in solving infinite game problems~\cite{Beyene:2014,BeyenePR13} 
can be specialized to the problem of realizability.  In this work, the authors 
describe a framework for analyzing arbitrary two-player games.  To 
provide proofs 
within the framework, {\em template formulas} must be provided by the user
that describe the shape of a Skolem function that is used to explicitly 
define an inductive invariant that demonstrates the realizability of a model.
Although this work is more general than ours, the applicability of the 
approach requires user-provided templates that are problem specific, so is 
not entirely automated.

The main contribution of our work is that it automatically checks the
realizability of infinite domain systems. The problem is, in general,
undecidable. Still, the application of bounded model checking can
still offer an approximate answer to the realizability problem as we
experienced by the fact that Z3 managed to solve the majority of our test
models.


\vspace{-0.1in}
\section{Conclusions and Future Work}
\label{sec:future_work}
\vspace{-0.1in}

In this paper, we have presented a new approach for determining realizability of contracts involving infinite theories using SMT solvers.  This approach allows analysis of a class of contracts that were previously not solvable using automated analysis.  The approach is both incomplete and conservative, i.e., it may return ``false positive'' results, declaring that a contract is not realizable when it could be realized.  However, it has been shown to be both fast and effective in practice on a variety of models.   

The results of this paper provide a good foundation towards further research in realizability.  In much the same way that many properties are not {\em inductive}, some contracts cannot be proven realizable using one step extensions.  We are examining alternate algorithms, similar to approaches such as IC3~\cite{bradley11}, which support property-directed invariant generation, to improve the approach presented here.  However, this requires generalizing the IC3 approach to solve quantified formulas (as well as to generalize counterexamples over quantified formulas).  We hope to demonstrate an approach involving a IC3-like algorithm in the near future.

In addition, for realizable systems, it is likely that we want to consider the {\em synthesis} problem, which we have not explicitly considered in this paper.  Synthesis aims to construct a concrete implementation of the contract, rather than determine its existence.  It is known for propositional systems that the synthesis problem is equivalent in complexity to the realizability problem~\cite{Pnueli89}, but it is not known (to us) whether this equivalence is true in the infinite-state case.


\vspace{-0.1in}
\subsubsection{Acknowledgments.}
This work was funded by DARPA and AFRL under contract FA8750-12-9-0179 (Secure Mathematically-Assured Composition of Control Models), and by NASA under contract NNA13AA21C (Compositional Verification of Flight Critical Systems), and by NSF under grant CNS-1035715 (Assuring the safety, security, and reliability of medical device cyber physical systems).
\vspace{-0.1in}

\bibliography{document}

\begin{thebibliography}{10}

\bibitem{SAE:AADL}
SAE-AS5506:
\newblock Architecture Analysis and Design Language.
\newblock SAE (2004)

\bibitem{Friedenthal08:sysml}
Friedenthal, S., Moore, A., Steiner, R.:
\newblock A Practical Guide to SysML: Systems Modeling Language.
\newblock Morgan Kaufmann Publishers Inc., San Francisco, CA, USA (2008)

\bibitem{AUTOSAR}
Consortium, A.:
\newblock Automotive Open System Architecture (AUTOSAR) Revision 4.2.1.
\newblock AUTOSAR (2014)

\bibitem{Gomez11:AADL}
Varona-Gomez, R., Villar, E.:
\newblock Aadl simulation and performance analysis in systemc.
\newblock In: Engineering of Complex Computer Systems, 2009 14th IEEE
  International Conference on. (2009)  323--328

\bibitem{Bozzano:2011:AADL}
Bozzano, M., Cimatti, A., Katoen, J.P., Nguyen, V.Y., Noll, T., Roveri, M.:
\newblock Safety, dependability and performance analysis of extended aadl
  models.
\newblock Comput. J. \textbf{54} (2011)  754--775

\bibitem{Apvrille13:security}
{A}pvrille, L., {R}oudier, Y.:
\newblock {S}ys{ML}-{S}ec: {A} model-driven environment for developing secure
  embedded systems.
\newblock In: {SAR}-{SSI} 2013, 8{\`e}me {C}onf{\'e}rence sur la
  {S}{\'e}curit{\'e} des {A}rchitectures {R}{\'e}seaux et des {S}yst{\`e}mes
  d'{I}nformation, 16-18 {S}eptembre 2013, {M}ont-de-{M}arsan, {F}rance,
  {M}ont-de-{M}arsan, {FRANCE} (2013)

\bibitem{Bozzano14:AADL}
Bozzano, M., Cimatti, A., Katoen, J.P., Katsaros, P., Mokos, K., Nguyen, V.Y.,
  Noll, T., Postma, B., Roveri, M.:
\newblock Spacecraft early design validation using formal methods.
\newblock Reliability Engineering and System Safety \textbf{132} (2014)

\bibitem{Whalen13:WhatHow:TwinPeaksIEEESoftware}
Whalen, M.W., Gacek, A., Cofer, D., Murugesan, A., Heimdahl, M.P., Rayadurgam,
  S.:
\newblock Your what is my how: Iteration and hierarchy in system design.
\newblock Software, IEEE \textbf{30} (2013)  54--60

\bibitem{rushby2011new}
Rushby, J.:
\newblock New challenges in certification for aircraft software.
\newblock In: Proceedings of the ninth ACM Int'l Conf. on Embedded software,
  ACM (2011)  211--218

\bibitem{Miller06:Shalls}
Miller, S.P., Tribble, A.C., Whalen, M.W., Heimdahl, M.P.E.:
\newblock Proving the shalls: Early validation of requirements through formal
  methods.
\newblock Int. J. Softw. Tools Technol. Transf. \textbf{8} (2006)  303--319

\bibitem{Hammond01:WiW}
Hammond, J., Rawlings, R., Hall, A.:
\newblock Will it work? [requirements engineering].
\newblock In: Requirements Engineering, 2001. Proceedings. Fifth IEEE Int'l
  Symposium on. (2001)  102 --109

\bibitem{NFM2012:CoGaMiWhLaLu}
Cofer, D.D., Gacek, A., Miller, S.P., Whalen, M.W., LaValley, B., Sha, L.:
\newblock Compositional verification of architectural models.
\newblock In Goodloe, A.E., Person, S., eds.: Proceedings of the 4th NASA
  Formal Methods Symposium (NFM 2012). Volume 7226., Berlin, Heidelberg,
  Springer-Verlag (2012)  126--140

\bibitem{Pnueli89}
Pnueli, A., Rosner, R.:
\newblock On the {S}ynthesis of a {R}eactive {M}odule.
\newblock Proceedings of the 16th ACM SIGPLAN-SIGACT symposium on Principles of
  Programming Languages (POPL'89) (1989)  179--190

\bibitem{Bohy12}
Bohy, A., Bruyère, V., Filiot, E., Jin, N., Raskin, J.F.:
\newblock Acacia+, a tool for {LTL} {S}ynthesis.
\newblock Proceedings of the 24th {I}nternational {C}onference on {C}omputer
  {A}ided {V}erification (CAV'12) (2012)  652--657

\bibitem{Hamza10}
Hamza, J., Jobstmann, B., Kuncak, V.:
\newblock Synthesis for {R}egular {S}pecifications over {U}nbounded {D}omains.
\newblock Proceedings of the 2010 Conference on Formal Methods in
  Computer-Aided Design (2010)  101--109

\bibitem{Chatterjee07}
Chatterjee, K., Henzinger, T.A.:
\newblock Assume-{G}uarantee {S}ynthesis.
\newblock Proceedings of the 13th International Conference on Tools and
  Algorithms for the Construction and Analysis of Systems (TACAS'07) (2007)
  261--275

\bibitem{hilt2013}
Murugesan, A., Whalen, M.W., Rayadurgam, S., Heimdahl, M.P.:
\newblock Compositional verification of a medical device system.
\newblock In: ACM Int'l Conf. on High Integrity Language Technology (HILT)
  2013, ACM (2013)

\bibitem{Katis:machine}
Katis, A., Gacek, A., Whalen, M.W.:
\newblock Machine-checked proofs for realizability checking algorithms (2015)
  Submitted \url{http://arxiv.org/abs/1502.01292}.

\bibitem{jkind}
Gacek, A.:
\newblock {JK}ind - a {J}ava implementation of the {KIND} model checker.
\newblock \url{https://github.com/agacek/jkind} (2014)

\bibitem{Hagen08:kind}
Hagen, G.:
\newblock Verifying safety properties of {L}ustre programs: an {SMT}-based
  approach.
\newblock PhD thesis, University of Iowa (2008)

\bibitem{DeMoura08:z3}
De~Moura, L., Bj{\o}rner, N.:
\newblock {Z3}: An efficient {SMT} solver.
\newblock In: Tools and Algorithms for the Construction and Analysis of
  Systems.
\newblock Springer (2008)  337--340

\bibitem{Murugesan2014}
Murugesan, A., Sokolsky, O., Rayadurgam, S., Whalen, M., Heimdahl, M., Lee, I.:
\newblock Linking {A}bstract {A}nalysis to {C}oncrete {D}esign: {A}
  {H}ierarchical {A}pproach to {V}erify {M}edical {CPS} {S}afety.
\newblock Proceedings of ICCPS'14 (2014)

\bibitem{Gunter00}
Gunter, C.A., Gunter, E.L., Jackson, M., Zave, P.:
\newblock A {R}eference model for {R}equirements and {S}pecifications.
\newblock IEEE Software \textbf{17} (2000)  37--43

\bibitem{Klein10}
Klein, U., Pnueli, A.:
\newblock Revisiting {S}ynthesis of {GR}(1) {S}pecifications.
\newblock Proceedings of the 6th International Conference on Hardware and
  Software: Verification and Testing (HVC'10) (2010)  161--181

\bibitem{Tini03}
Tini, S., Maggiolo-Schettini, A.:
\newblock Compositional {S}ynthesis of {G}eneralized {M}ealy {M}achines.
\newblock Fundamenta Informaticae \textbf{60} (2003)  367--382

\bibitem{Benes12}
Beneš, N., Černá, I., Štefaňák, F.:
\newblock Factorization for {C}omponent-{I}nteraction {A}utomata.
\newblock Proceedings of the 38th International Conference on Current Trends in
  Theory and Practice of Computer Science (2012)  554--565

\bibitem{Aziz95}
Aziz, A., Balarin, F., Braton, R., Sangiovanni-Vincentelli, A.:
\newblock Sequential {S}ynthesis using {SIS}.
\newblock Proceedings of the 1995 IEEE/ACM International Conference on
  Computer-Aided Design (ICCAD'95) (1995)  612--617

\bibitem{micheli_aaai_2012}
Cimatti, A., Micheli, A., Roveri, M.:
\newblock Solving temporal problems using {SMT}: Weak controllability.
\newblock In: AAAI. (2012)  448--454

\bibitem{micheli_cp_2012}
Cimatti, A., Micheli, A., Roveri, M.:
\newblock Solving temporal problems using {SMT}: Strong controllability.
\newblock In: CP. (2012)  248--264

\bibitem{micheli_constraints_2014}
Cimatti, A., Micheli, A., Roveri, M.:
\newblock Solving strong controllability of temporal problems with uncertainty
  using {SMT}.
\newblock Constraints (2014)

\bibitem{Beyene:2014}
Beyene, T., Chaudhuri, S., Popeea, C., Rybalchenko, A.:
\newblock A constraint-based approach to solving games on infinite graphs.
\newblock In: Proceedings of the 41st ACM SIGPLAN-SIGACT Symposium on
  Principles of Programming Languages. POPL '14, New York, NY, USA, ACM (2014)
  221--233

\bibitem{BeyenePR13}
Beyene, T.A., Popeea, C., Rybalchenko, A.:
\newblock Solving existentially quantified {H}orn clauses.
\newblock In: Computer Aided Verification - 25th International Conference,
  {CAV} 2013, Saint Petersburg, Russia, July 13-19, 2013. Proceedings. (2013)
  869--882

\bibitem{bradley11}
Bradley, A.:
\newblock {SAT}-based model checking without unrolling.
\newblock VMCAI (2011)

\end{thebibliography}
\bibliographystyle{splncs}

\end{document}